\documentclass[notitle,showpacs,twocolumn]{revtex4}
\usepackage{graphicx}
\usepackage{tikz-cd}
\usepackage{dcolumn}
\usepackage{amsmath}
\usepackage{amssymb}
\usepackage{mathtools}
\usepackage{hyperref}

\newtheorem{theorem}{Theorem}[section]

\newtheorem{proposition}[theorem]{Proposition}

\newtheorem{definition}[theorem]{Definition}
\newenvironment{proof}[1][Proof]{\begin{trivlist}
\item[\hskip \labelsep {\bfseries #1}]}{\end{trivlist}}

\newcommand{\qed}{\nobreak \ifvmode \relax \else
	\ifdim\lastskip<1.5em \hskip- \lastskip
	\hskip 0.5em plus0em minus0.5em \fi \nobreak
	\vrule height0.75em width0.5em depth0.25em\fi}

\begin{document}

\title{Non-existence of expansion-free dynamical stars with rotation and spatial twist}

\author{Abbas \surname{Sherif}}
\email{abbasmsherif25@gmail.com}
\affiliation{Astrophysics and Cosmology Research Unit, School of Mathematics, Statistics and Computer Science, University of KwaZulu-Natal, Private Bag X54001, Durban 4000, South Africa}

\author{Rituparno \surname{Goswami}}
\email{vitasta9@gmail.com}
\affiliation{Astrophysics and Cosmology Research Unit, School of Mathematics, Statistics and Computer Science, University of KwaZulu-Natal, Private Bag X54001, Durban 4000, South Africa}

\author{Sunil \surname{Maharaj}}
\email{maharaj@ukzn.ac.za}
\affiliation{Astrophysics and Cosmology Research Unit, School of Mathematics, Statistics and Computer Science, University of KwaZulu-Natal, Private Bag X54001, Durban 4000, South Africa}

\begin{abstract}
Extending a previous work by the same authors, we investigate the existence of expansion-free dynamical stars with non-zero spatial twist and rotation and show that such stars cannot exist. Firstly, it is shown that a rotating expansion-free dynamical star with zero twist cannot exist. This is due to the fact that such stars cannot radiate and they are shear-free, in which case the energy density \(\rho\) is time independent. Secondly, we prove that a non-rotating expansion-free dynamical star with non-zero spatial twist also cannot exist, as either the strong energy condition must be violated, i.e. \(\rho+3p<0\), or the star must be shear-free in which case the star is static (\(\Theta=\Omega=\Sigma=0\)). Finally, if we insist that the rotation and spatial twist are simultaneously non-zero, then the star cannot be shear-free in which case we obtain a quadratic polynomial equation in \(\phi\) and \(\Sigma\) with no real solutions. Therefore such stars cannot exist.
\end{abstract}
\maketitle

\section{Introduction}

Models of radiating stars in general relativity play a central role in the study of gravitational collapse and the astrophysics of gravitating bodies. Physically relevant exact models were obtained by Tewari and Charan \cite{tac1}, Tewari \cite{tac2}, and Ivanov \cite{iv1,iv2,iv3}. These examples provide interesting insights into the processes involved during stellar evolution. It has also been anisotrophy and dissipative effects during gravitational collapse have influence on the collapse rate and temperature profiles in radiating stars by Reddy \textit{et al.} \cite{red1}. There are classes of exact solutions to the Einstein's field equations (EFEs) obtained, referred to as Euclidean stars, which have been shown to regain Newtonian stars within the appropriate limit \cite{sant1,gov1,gov2}. In recent years the method of Lie analysis of differential equations using symmetry invariance has proved an invaluable and systematic tool in obtaining general categories of exact solutions to the boundary condition of radiating stellar objects \cite{gov3,moh1,moh2}. There is an important class of radiating stars that has been introduced by Herrera \textit{et al.} \cite{lh1} which are expansion-free. Expansion-free dynamical models implies the existence of a cavity or void. One important feature of expansion-free models is that matter distributions with a vanishing expansion scalar have to be inhomogeneous. These physical features should have important astrophysical consequences for spherically symmetric distributions. Also, such radiating astrophysical model might offer a plausible explanation into the existence of voids that have been observed on the cosmological scales. Various authors have explored expansion-free dynamical models with different considerations. Some of the studies containing the description of physical properties of expansion-free dynamical radiating stars can be found in several works \cite{lh5,kum1,kum2}. The peak in interest in these models laid in the fact that in such models is the possibility that they could help explain the existence of voids on cosmological scales. In 2008, Herrera and co-authors \cite{lh1} studied such models with non-zero shear and showed that the appearance of a cavity (see reference \cite{pp1} for more discussion), with matter which is anisotropic and dissipative, undergoing explosion is inevitable. The same authors followed this result by a 2009 paper in \cite{lh2} in which they ruled out the Skripkin expansion-free dynamcal model (see reference \cite{skr1}) with constant energy density and isotropic pressure. Another study in \cite{lh3} involved the study of models collapsing adiabatically, and showed that the instability was independent of the star's stiffness. In particular, it was shown that the instability is entirely governed by the pressures and the radial profile of the energy density. In a recent work by Sherif \textit{et al.} \cite{asr1} the authors employed, for the first time, the \(1+1+2\) formalism (a semitetrad covariant method for anlyzing the field equations) to study the properties of expansion free models. With emphasis on non-rotating and non-twisting stars the authors found that a necessary condition for the existence of such stars is that the star simultaneously accelerate and radiate. It was also shown in the same paper that these stars must possess a conformally flat geometry. 

In this paper we study the required geometric and thermodynamic properties for the existence of a relativistic expansion-free dynamical star with at least the rotation or spatial twist being non-zero. This analysis falls in the scope of stability analysis of self-gravitating systems (some of the references are given in \cite{mr1,jw1,nos1,gmm1,rittt1}). Our approach will be to fix either of the rotation or the spatial twist to zero and see whether indeed expansion-free dynamical models of such stars can exist. In particular we would like to know the restrictions that the addition of spatial twist or/and rotation induces on the geometric and matter quanties such as acceleration, heat flux, e.t.c. As with our previous work \cite{asr1}, we will make use of the equivalent forms of the field equations from the \(1+1+2\) semi-tetrad covariant formulation of general relativity \cite{pg1,ts1,ts2,ts3,cc1,rit1}. The semi-tetrad formalism has proven to be an extremely useful approach in displaying geometrical features in a transparent fashion which are generally very difficult to find using other approaches.

In section \ref{02} we briefly introduce the \(1+1+2\) semi-tetrad formalism and provide a definition of locally rotationally symmetric (LRS) spacetimes. In section \ref{03} we present the results of the paper, a complete analysis of the expansion-free model with rotation and spatial twist. We conclude with a discussion of the results in section \ref{04}.

\section{Locally rotationally symmetric spacetimes and its \(1+1+2\) semitetrad splitting}\label{02}

We provide some background material in this section, covering the \(1+1+2\) semi-tetrad covariant formalism as well as notes on and calculations of useful quantities, utilized in this paper. 

Stellar models that are rotating and twisting can be studied using the models of spacetimes known as \textit{locally rotationally symmetric} spacetimes  \cite{ggff1,sge1}. As such we will use this model here to investigate expansion free-dynamical stars that are either rotating or possessing spatial twist or both. We start by explicitly defining LRS spacetimes.
\begin{definition}
A \textbf{locally rotationally symmetric (LRS)} spacetime is a spacetime in which at each point \(p\in M\), there exists a continuous isotropy group generating a multiply transitive isometry group on \(M\) \cite{crb1,gbc1,sge1,vege1, ssgos1,ssgos2}. The general metric of LRS spacetimes is given by

\begin{eqnarray}\label{jan29191}
\begin{split}
ds^2&=-A^2dt^2 + B^2d\chi^2 + F^2 dy^2 \\
&+ \left[\left(F\bar{D}\right)^2+ \left(Bh\right)^2 - \left(Ag\right)^2\right]dz^2\\ 
&+ \left(A^2gdt - B^2hd\chi\right)dz,
\end{split}
\end{eqnarray}
where \(A^2,B^2,F^2\) are functions of \(t\) and \(\chi\), \(\bar{D}^2\) is a function of \(y\) and \(k\) (\(k\) fixes the geometry of the \(2\)-surfaces), and \(g,h\) are functions of \(y\). 
\end{definition}
In the limiting case that \(g=h=0\) we recover the well known spherically symmetric LRS II class of spacetimes which generalizes spherically symmetric solutions to Einstein field equations (EFEs). These spacetimes with vanishing rotation and spatial twist was employed in \cite{asr1} to study expansion-free and dynamic stellar models. LRS spacetimes, a generalization of LRS II spacetimes, on the other hand include solutions with nonzero vorticity and nonzero spatial twist. Some of these solutions include the G\"{o}del's world model, the Kantowski-Sachs models and the Bianchi models, invariant under the \(G_3\) groups of types I, II, VIII and IX (see for example the reference \cite{gft2}). In fact the G\"{o}del's world model, a famous albeit unphysical solution to the field equations \cite{}, is an expansion-free model that is rotating with zero spatial twist. Properties of such expansion-free dynamical stars will be investigated in section \ref{sub:1} and properties necessary for their existence will be determined.

Let us next introduce the \(1+1+2\) covariant splitting of spacetime and the resulting fields equations for LRS spacetimes, as well as derivatives of the unit vector fields \cite{cc1,rit1}. 

To start with, let (\(M,g_{ab}\)) be a spacetime manifold, with associated metric tensor \(g_{ab}\). To any timelike congruence of an observer we may associate a unit vector field \(u^a\) tangent to the congruence which satisfies \(u^au_a=-1\). One may then split \(M\) as follows: Given any \(4\)-vector \(U^a\) in the spacetime, the projection tensor \(h_a^{\ b}\equiv g_a^{\ b}+u_au^b\), projects \(U^a\) onto the \(3\)-space as
\begin{eqnarray*}
U^a&=&Uu^a + U^{\langle a \rangle },
\end{eqnarray*}
where \(U\) is the scalar along \(u^a\) and \(U^{\langle a \rangle }\) is the projected \(3\)-vector \cite{ggff2}. The splitting splits \(g_{ab}\) into components associated with the \(u^a\) and spatial directions. This naturally gives rise to two derivatives:
\begin{itemize}
\item The \textit{covariant time derivative} (or simply the dot derivative)  along the observers' congruence. Given any tensor \(S^{a..b}_{\ \ \ \ c..d}\), we have \(\dot{S}^{a..b}_{\ \ \ \ c..d}\equiv u^e\nabla_eS^{a..b}_{\ \ \ \ c..d}\).

\item Fully orthogonally \textit{projected covariant derivative} \(D\) with the tensor \(h_{ab}\), with the total projection carried out on all the free indices. Given any tensor \(S^{a..b}_{\ \ \ \ c..d}\), we have \(D_eS^{a..b}_{\ \ \ \ c..d}\equiv h^a_{\ f}h^p_{\ c}...h^b_{\ g}h^q_{\ d}h^r_{\ e}\nabla_rS^{f..g}_{\ \ \ \ p..q}\).
\end{itemize}
This \(1+3\) splitting of the spacetime irreducibly splits the covariant derivative of \(u^a\) as
\begin{eqnarray}\label{mmmn}
\nabla_au_b=-A_au_b+\frac{1}{3}h_{ab}\Theta+\sigma_{ab}.
\end{eqnarray}
In \eqref{mmmn}, the vector \(A_a=\dot{u}_a\) is the acceleration vector, \(\Theta\equiv D_au^a\) - the trace of the fully orthogonally projected covariant derivative of \(u^a\) - is the expansion and \(\sigma_{ab}=D_{\langle b}u_{a\rangle}\) is the shear tensor (wherever used in this paper, angle brackets will denote the projected symmetric trace-free part of the tensor). In the particular case of LRS spacetimes, all vector and tensor quantities vanish identically (see reference \cite{cc1} for details). 

The splitting further allows for the energy momentum tensor to be decomposed as
\begin{eqnarray}
T_{ab}=\rho u_au_b + 2q_{(a}u_{b)} +ph_{ab} + \pi_{ab},
\end{eqnarray}
where \(\rho\equiv T_{ab}u^au^b\) is the energy density, \(q_a=-h_a^{\ c}T_{cd}u^d\) is the \(3\)-vector defining the heat flux, \(p\equiv\left(1/3\right)h^{ab}T_{ab}\) is the isotropic pressure and \(\pi_{ab}\) is the anisotropic stress tensor.

If there is a preferred unit normal spatial direction \(e^a\) which is the case with LRS II spacetimes, the metric \(g_{ab}\) can be split into terms along the \(u^a\) and \(e^a\) directions (the vector field \(e^a\) splits the \(3\)-space), as well as on the \(2\)-surface, i.e. 
\begin{eqnarray}
g_{ab}=N_{ab}-u_au_b+e_ae_b,
\end{eqnarray} 
where the projection tensor \(N_{ab}\) projects any two vector orthogonal to \(u^a\) and \(e^a\) onto the \(2\)-surface defined by the sheet (\(N^{\ \ a}_a=2, u^aN_{ab}=0, \ e^aN_{ab}=0\)), and \(e^a\) is defined such that \(e^ae_a=1\) and it is orthogonal to \(u^a\), i.e. \(u^ae_a=0\). This is referred to as the \(1+1+2\) splitting. This splitting of the spacetime additionally gives rise to the splitting of the covariant derivatives along the \(e^a\) direction and on the \(2\)-surface:
\begin{itemize}
\item The \textit{hat derivative} is the spatial derivative along the vector field \(e^a\): Given a \(3\)-tensor \(\psi_{a..b}^{\ \ \ \ c..d}\), we have \(\hat{\psi}_{a..b}^{\ \ \ \ c..d}\equiv e^fD_f\psi_{a..b}^{\ \ \ \ c..d}\).

\item The \textit{delta derivative} is the projected spatial derivative on the \(2\)-sheet (projection by the tensor \(N_a^{\ b}\)), and the projection is carried out on all the free indices: Given any \(3\)-tensor \(\psi_{a..b}^{\ \ \ \ c..d}\), we have \(\delta_e\psi_{a..b}^{\ \ \ \ c..d}\equiv N_a^{\ f}..N_b^{\ g}N_h^{\ c}..N_i^{\ d}N_e^{\ j}D_j\psi_{f..g}^{\ \ \ \ h..i}\).
\end{itemize} 

The complete set of \(1+1+2\) covariant scalars fully describing the LRS class of spacetimes are \cite{cc1}
\begin{eqnarray*}
\lbrace{A,\Theta,\phi, \Sigma, \mathcal{E}, \mathcal{H}, \rho, p, \Pi, Q, \Omega, \xi\rbrace}. 
\end{eqnarray*}
The quantity \(\phi\equiv\delta_ae^a\) is the sheet expansion, \(\Sigma\equiv\sigma_{ab}e^ae^b\) is the scalar associated with the shear tensor \(\sigma_{ab}\), \(\mathcal{E}\equiv E_{ab}e^ae^b\) is the scalar associated with the electric part of the Weyl tensor \(E_{ab}\), \(\mathcal{H}\equiv H_{ab}e^ae^b\) is the scalar associated with the magnetic part of the Weyl tensor \(\mathcal{H}_{ab}\), \(\Pi\equiv\pi_{ab}e^ae^b\) is the anisotropic stress scalar, \(Q\equiv -e^aT_{ab}u^b=q_ae^a\) is the scalar associated to the heat flux vector \(q_a\). The quantities \(\xi\) and \(\Omega\) are the spatial twist and rotation scalar respectively, which are defined by \(\xi=\left(1/2\right)\varepsilon^{ab}\delta_ae_b\) (where \(\varepsilon_{ab}\equiv \varepsilon_{abc}e^c= u^d\eta_{dabcd}e^c\) is the levi civita \(2\)-tensor, the volume element of the \(2\)-surface) and \(\Omega=e^a\omega_a\) (where \(\omega^a=\Omega e^a + \Omega^a\) is the rotation vector, with \(\Omega^a\) being the sheet component of \(\omega^a\)).

The full covariant derivatives of the vector fields \(u^a\) and \(e^a\) are given by \cite{cc1}
\begin{subequations}\label{4}
\begin{align}
\nabla_au_b&=-Au_ae_b + e_ae_b\left(\frac{1}{3}\Theta + \Sigma\right) \\
&+ N_{ab}\left(\frac{1}{3}\Theta -\frac{1}{2}\Sigma\right),\label{4}\\
\nabla_ae_b&=-Au_au_b + \left(\frac{1}{3}\Theta + \Sigma\right)e_au_b +\frac{1}{2}\phi N_{ab}.\label{444}
\end{align}
\end{subequations}
We also note the useful expression 
\begin{eqnarray}\label{redpen}
\hat{u}^a&=&\left(\frac{1}{3}\Theta+\Sigma\right)e^a.
\end{eqnarray}

Any given scalar \(\psi\) satisfies the commutation relation

\begin{eqnarray}\label{ghh1}
\hat{\dot{\psi}}-\hat{\dot{\psi}}=-A\dot{\psi}+\left(\frac{1}{3}\Theta+\Sigma\right)\hat{\psi}.
\end{eqnarray}
We will utilize this relation throughout this work when seeking constraint equations. The field equations for LRS spacetimes are given as propagation and evolution of the covariant scalars \cite{cc1}:

\begin{itemize}

\item \textit{Evolution}
\begin{subequations}
\begin{align}
\frac{2}{3}\dot{\Theta}-\dot{\Sigma}&=A\phi- \frac{1}{2}\left(\frac{2}{3}\Theta-\Sigma\right)^2 - 2\Omega^2 + \mathcal{E} - \frac{1}{2}\Pi \notag \\
& - \frac{1}{3}\left(\rho+3p\right),\label{sube1}\\
\dot{\phi}&=\left(\frac{2}{3}\Theta-\Sigma\right)\left(A-\frac{1}{2}\phi\right) + 2\xi\Omega + Q,\label{sube2}\\
\dot{\xi}&=-\frac{1}{2}\left(\frac{2}{3}\Theta-\Sigma\right)\xi + \left(A-\frac{1}{2}\phi\right)\Omega,\label{sube3}\\
\dot{\Omega}&=A\xi-\left(\frac{2}{3}\Theta-\Sigma\right)\Omega,\label{sube4}\\
\dot{\mathcal{H}}&=-3\xi\mathcal{E}-\frac{3}{2}\left(\frac{2}{3}\Theta-\Sigma\right)\mathcal{H}+\Omega Q,\label{sube5}\\
\dot{\mathcal{E}}-\frac{1}{3}\dot{\rho}+\frac{1}{2}\dot{\Pi}&=-\left(\frac{2}{3}\Theta-\Sigma\right)\left(\frac{3}{2}\mathcal{E}+\frac{1}{4}\Pi\right)+\frac{1}{2}\phi Q\notag\\
&+3\xi\mathcal{H}+\frac{1}{2}\left(\frac{2}{3}\Theta-\Sigma\right)\left(\rho+p\right),\label{sube6}
\end{align}
\end{subequations}
\item \textit{Propagation}
\begin{subequations}
\begin{align}
\frac{2}{3}\hat{\Theta}-\hat{\Sigma}&=\frac{3}{2}\phi\Sigma + 2\xi\Omega + Q,\label{sube7}\\
\hat{\phi}&=-\frac{1}{2}\phi^2 + \left(\frac{1}{3}\Theta+\Sigma\right)\left(\frac{2}{3}\Theta-\Sigma\right)+2\xi^2\notag\\
&-\frac{2}{3}\rho-\mathcal{E}-\frac{1}{2}\Pi,\label{sube8}\\
\hat{\xi}&=-\phi\xi + \left(\frac{1}{3}\Theta+\Sigma\right)\Omega,\label{sube9}\\
\hat{\Omega}&=\left(A-\phi\right)\Omega,\label{sube10}\\
\hat{\mathcal{H}}&=-\left(3\mathcal{E}+\rho+p-\frac{1}{2}\Pi\right)\Omega-3\phi\mathcal{H}\notag\\
&-Q\xi,\label{sube11}\\
\hat{\mathcal{E}}-\frac{1}{3}\hat{\rho}+\frac{1}{2}\hat{\Pi}&=-\frac{3}{2}\phi\left(\mathcal{E}+\frac{1}{2}\Pi\right)-\frac{1}{2}\left(\frac{2}{3}\Theta-\Sigma\right)Q\notag\\
&+3\Omega\mathcal{H}\label{sube12}
\end{align}
\end{subequations}
\item \textit{Evolution/Propagation}
\begin{subequations}
\begin{align}
\hat{A}-\dot{\Theta}&=-\left(A+\phi\right)A-\frac{1}{3}\Theta^2+\frac{3}{2}\Sigma^2-2\Omega^2\notag\\
&+\frac{1}{2}\left(\rho+3p\right),\label{sube13}\\
\dot{\rho}+\hat{Q}&=-\Theta\left(\rho+p\right)-\left(2A+\phi\right)Q-\frac{3}{2}\Sigma\Pi,\label{sube14}\\
\dot{Q}+\hat{p}+\hat{\Pi}&=-\left(A+\frac{3}{2}\phi\right)\Pi-\left(\frac{4}{3}\Theta+\Sigma\right)Q\notag\\
&-\left(\rho+p\right)A,\label{sube15}
\end{align}
\end{subequations}
\item \textit{Constraint}
\begin{eqnarray}\label{ggh10}
\mathcal{H}=3\Sigma\xi-\left(2\mathcal{A}-\phi\right)\Omega.
\end{eqnarray}

\end{itemize}
Let us now analyze the expansion-free dynamical models with rotation and spatial twist.

\section{Results}\label{03}

In \cite{asr1} we considered expansion-free dynamical stars that are non-rotating and non-twisting. It was shown that the existence of such models requires the star to simultaneously accelerate and radiate, in which case the star is necessarily conformally flat. We consider here the case in which at least either one of \(\Omega\) or \(\xi\) is non-vanishing. Thus we are considering the following three cases \cite{vege1,sge1,ssgos1}: 

\begin{itemize}
\item[1.] \(\xi=0;\Omega\neq 0\): These models fall under the class of spacetimes known as LRS I spacetimes, with \(e^a\) hypersurface orthogonal and \(u^a\) twisting. A well know example is the G\(\ddot{\text{o}}\)del solution. 
\item[2.] \(\xi\neq 0;\Omega= 0\): These models fall under the class of spacetimes known as LRS III spacetimes, with \(e^a\) twisting and \(u^a\) hypersurface orthogonal. 
\item[3.] \(\xi\neq 0;\Omega\neq 0\): These models, investigated in \cite{ssgos1}, have the property that the heat flux \(Q\) cannot be zero and specific energy conditions are required to be satisfied, i.e.

\begin{eqnarray}\label{gg}
-\frac{1}{2}\left(\rho+p+\Pi\right)<Q<\frac{1}{2}\left(\rho+p+\Pi\right).
\end{eqnarray} 
\end{itemize}
One therefore expects that an expansion-free dynamical model to exist in models with \(\xi\neq 0\) and \(\Omega\neq 0\). The star being dynamical implies that all of the thermodynamic quantities, including \(p,\rho,\Pi, Q\) e.t.c., are functions of time.

\subsection{Case 1: \(\xi=0;\Omega\neq 0\)}\label{sub:1}

Let us start by considering the case of a rotating expansion-free star with no spatial twist. From \eqref{sube9} we have

\begin{eqnarray}\label{ggh11}
0=\Sigma\Omega.
\end{eqnarray}
Since by assumption, \(\Omega\neq 0\), we must have \(\Sigma=0\). Furthermore, from \eqref{sube3} we obtain
\begin{eqnarray}\label{ggh12}
0=\left(A-\frac{1}{2}\phi\right)\Omega,
\end{eqnarray}
which, from \eqref{ggh10}, gives \(\mathcal{H}=0\), so that for such stars the Weyl tensor is purely electric. Using \eqref{sube5} one has

\begin{eqnarray}\label{ggh13}
0=\Omega Q,
\end{eqnarray}
from which we obtain \(Q=0\). Therefore the star is not dynamical as the energy density is time independent, i.e. \(\dot{\rho}=0\) from \eqref{sube14}. It is also not difficult to show that such stars will necessarily accelerate. To see this, assume \(A=0\). Then from \eqref{ggh12}, since by assumption \(\Omega\neq 0\) we must have \(\phi=0\) as well. From \eqref{sube1}, \eqref{sube8}, \eqref{sube11} and \eqref{sube13} we obtain the constraints

\begin{subequations}
\begin{align}
0&=-2\Omega^2+\mathcal{E}-\frac{1}{3}\left(\rho+3p\right)-\frac{1}{2}\Pi,\label{sube16}\\
0&=\mathcal{E}+\frac{2}{3}\rho+\frac{1}{2}\Pi,\label{sube17}\\
0&=3\mathcal{E}+\rho+p-\frac{1}{2}\Pi,\label{sube18}\\
0&=-2\Omega^2+\frac{1}{2}\left(\rho+3p\right).\label{sube19}
\end{align}
\end{subequations}
Comparing \eqref{sube17} and \eqref{sube18} we obtain

\begin{eqnarray}\label{ggh14}
0=-\rho+p-2\Pi.
\end{eqnarray}
Substituting \eqref{ggh14} and \eqref{sube17} into \eqref{sube16} we obtain the constraint

\begin{eqnarray}\label{ggh15}
0=-2\Omega^2-\frac{1}{2}\left(\rho+3p\right),
\end{eqnarray}
which upon comparing to \eqref{sube19} gives 

\begin{eqnarray}\label{ggh15}
0=\rho+3p.
\end{eqnarray}
Therefore \(\Omega=0\) (from either \eqref{sube19} or \eqref{ggh15}), which contradicts the assumption that \(\Omega\neq 0\). Hence we have \(A\neq 0\). In summary,

\begin{theorem}\label{th2}
There cannot exist an expansion-free dynamical star with vanishing spatial twist and non-zero rotation.
\end{theorem}

Though these stars are not dynamical, we have enumerated several properties we expect such stars to have. In particular, the star is shear-free and accelerates without radiating.

\subsection{Case 2: \(\xi\neq 0;\Omega= 0\)}

Next, we consider non-rotating expansion-free dynamical stars with non-zero spatial twist. We state and prove the following

\begin{theorem}\label{th2}
There cannot exist an expansion-free dynamical star with vanishing rotation and non-zero spatial twist.
\end{theorem}
\begin{proof}
To prove this, we will show that if such star is to exist, then the star will either be static or it will violate the strong energy condition (SEC). From \eqref{sube4} we have

\begin{eqnarray}\label{ggh16}
0=A\xi,
\end{eqnarray}
so we must have \(A=0\) since by assumption \(\xi\neq 0\). Using \eqref{sube13} we have

\begin{eqnarray}\label{ggh17}
\Sigma^2=-\frac{1}{3}\left(\rho+3p\right).
\end{eqnarray}
Thus for such a star to exist we must have \(\rho+3p<0\), except in the case that the star is shear-free, in which case the star is static (\(\Omega=\Theta=\Sigma=0\)).\qed
\end{proof}
In fact in this case we have shown that even the expansion-free condition cannot hold, not only that it is not dynamical.

\subsection{Case 3: \(\xi\neq 0;\Omega\neq 0\)}

Finally, we consider the case of a simultaneously rotating and twisting expansion-free dynamical star. We start by taking the dot derivative of \eqref{sube1} and the hat derivative of \eqref{sube7} and obtain respectively

\begin{subequations}
\begin{align}
-\hat{\dot{\Sigma}}&=\phi\hat{A}+A\hat{\phi}-\Sigma\hat{\sigma}-4\Omega\hat{\Omega}+\hat{\mathcal{E}}-\frac{1}{2}\hat{\Pi}-\frac{1}{3}\hat{\rho}-\hat{p}\notag\\
&=-A^2\phi-\frac{3}{2}A\phi^2-A\left(\frac{2}{3}\rho+\mathcal{E}+\frac{1}{2}\Pi\right)+3\phi\Sigma^2\notag\\
&+2\phi\Omega^2+\frac{1}{2}\phi\left(\rho+3p\right)+\frac{3}{2}\Sigma Q+2\Omega\Sigma\xi+3\Omega\mathcal{H}\notag\\
&-\frac{3}{2}\phi\left(\mathcal{E}+\frac{1}{2}\Pi\right)-\left(\hat{p}+\hat{\Pi}\right),\label{sube20}\\
-\dot{\hat{\Sigma}}&=\frac{3}{2}\Sigma\dot{\phi}+\frac{3}{2}\phi\dot{\Sigma}+2\Omega\dot{\xi}+2\xi\dot{\Omega}+\dot{Q}\notag\\
&=-\frac{3}{2}A\Sigma^2+\frac{3}{2}\phi\Sigma^2+6\Omega\Sigma\xi+\frac{3}{2}\Sigma Q-\frac{3}{2}A\phi^2\notag\\
&+2\phi\Omega^2-\frac{3}{2}\phi \mathcal{E}+\frac{3}{4}\phi\Pi+\frac{1}{2}\phi\left(\rho+3p\right)+2A\Omega^2\notag\\
&+2A\xi^2+\dot{Q}.\label{sube21}
\end{align}
\end{subequations}
Taking the difference of \eqref{sube20} and \eqref{sube21} and using \eqref{sube15} we obtain 

\begin{eqnarray}\label{ggh18}
\begin{split}
-\hat{\dot{\Sigma}}+\dot{\hat{\Sigma}}&=-A^2\phi+\frac{1}{3}A\left(\rho+3p\right)-A\left(\mathcal{E}-\frac{1}{2}\Pi\right)\\
&+\frac{3}{2}\phi\Sigma^2-4\Omega\Sigma\xi-6A\Omega^2+\frac{1}{2}A\Sigma^2+\Sigma Q\\
&+3\Omega\mathcal{H}.
\end{split}
\end{eqnarray}
Using the commutation relation in \eqref{ghh1} on \(\Sigma\) we have

\begin{eqnarray}\label{ggh19}
\begin{split}
-\hat{\dot{\Sigma}}+\dot{\hat{\Sigma}}&=-A\dot{\Sigma}+\Sigma\hat{\Sigma}\\
&=-A^2\phi+\frac{1}{2}A\Sigma^2+2A\Omega^2-A\left(\mathcal{E}-\frac{1}{2}\Pi\right)\\
&+\frac{1}{3}A\left(\rho+3p\right)+\frac{3}{2}\phi\Sigma^2+2\Omega\Sigma\xi+\Sigma Q.
\end{split}
\end{eqnarray}
Comparing \eqref{ggh18} and \eqref{ggh19} and using \eqref{ggh10} we obtain the constraint

\begin{eqnarray}\label{ggh20}
\left(\frac{14}{3}A-\phi\right)\Omega=\Sigma\xi.
\end{eqnarray}
Now, taking the dot derivative of \eqref{sube2} and the hat derivative of \eqref{sube8} we obtain respectively

\begin{subequations}
\begin{align}
\hat{\dot{\phi}}&=-\Sigma\hat{A}-A\hat{\Sigma}+\frac{1}{2}\Sigma\hat{\phi}+\frac{1}{2}\phi\hat{\Sigma}+2\Omega\hat{\xi}+2\xi\hat{\Omega}+\hat{Q}\notag\\
&=A^2\Sigma+\frac{5}{2}A\phi\Sigma-2\Sigma^3+4\Sigma\Omega^2-\frac{5}{6}\Sigma\rho-\frac{3}{2}\Sigma p\notag\\
&+4A\Omega\xi-\Sigma\phi^2+\Sigma\xi^2+\frac{1}{2}\left(2A-\phi\right)Q-5\phi\Omega\xi\notag\\
&-\frac{1}{2}\Sigma\mathcal{E}-\frac{1}{4}\Sigma\Pi+\hat{Q},\label{sube22}\\
\dot{\hat{\phi}}&=-\phi\dot{\phi}-2\Sigma\dot{\Sigma}+4\xi\dot{\xi}-\frac{2}{3}\dot{\rho}-\left(\dot{\mathcal{E}}+\frac{1}{2}\dot{\Pi}\right)\notag\\
&=3A\phi\Sigma-\Sigma^3-\frac{1}{2}\Sigma\phi^2-4\phi\Omega\xi-\frac{3}{2}\phi Q+\frac{1}{2}\Sigma\mathcal{E}\notag\\
&-4\Sigma\Omega^2+2\Sigma\xi^2-\frac{5}{4}\Sigma\Pi+4A\Omega\xi-\frac{1}{6}\Sigma\rho-\frac{3}{2}\Sigma p\notag\\
&-3\xi\mathcal{H}-\dot{\rho}.\label{sube23}
\end{align}
\end{subequations}
Taking the difference of \eqref{sube22} and \eqref{sube23} and using \eqref{sube14} we obtain 

\begin{eqnarray}\label{ggh21}
\begin{split}
\hat{\dot{\phi}}-\dot{\hat{\phi}}&=A^2\Sigma-\frac{1}{2}A\phi\Sigma-\Sigma^3+8\Sigma\Omega^2-\frac{2}{3}\Sigma\rho-\frac{1}{2}\Sigma\phi^2\\
&-\Sigma\xi^2-AQ-\phi\Omega\xi-\Sigma\mathcal{E}+\Sigma\Pi+3\xi\mathcal{H}.
\end{split}
\end{eqnarray}
Using the commutation relation in \eqref{ghh1} on \(\phi\) we have

\begin{eqnarray}\label{ggh22}
\begin{split}
\hat{\dot{\phi}}-\dot{\hat{\phi}}&=-A\dot{\phi}+\Sigma\hat{\phi}\\
&=A^2\Sigma-\frac{1}{2}A\phi\Sigma-2A\Omega\xi-AQ-\frac{1}{2}\Sigma\phi^2-\Sigma^3\\
&+2\Sigma\xi^2-\frac{2}{3}\Sigma\rho-\Sigma\mathcal{E}-\frac{1}{2}\Sigma\Pi.
\end{split}
\end{eqnarray}
Comparing \eqref{ggh21} and \eqref{ggh22} and using \eqref{ggh10} we obtain the constraint

\begin{eqnarray}\label{ggh23}
\left(9\Sigma\xi+\phi\Omega-5A\Omega\right)\xi+\left(\frac{3}{2}\Pi+8\Omega^2\right)\Sigma=0.
\end{eqnarray}
Next, taking the dot derivative of \eqref{sube3} and the hat derivative of \eqref{sube9} we obtain respectively

\begin{subequations}
\begin{align}
\hat{\dot{\xi}}&=\frac{1}{2}\xi\hat{\Sigma}+\frac{1}{2}\Sigma\hat{\xi}+\Omega\left(\hat{A}-\frac{1}{2}\hat{\phi}\right)+\left(A-\frac{1}{2}\phi\right)\hat{\Omega}\notag\\
&=-\frac{5}{4}\phi\Sigma\xi-2\Omega\xi^2-\frac{1}{2}\xi Q + \frac{5}{2}\Omega\Sigma^2-\frac{5}{2}A\phi\Omega\notag\\
&+\frac{3}{4}\Omega\phi^2-2\Omega^3+\frac{5}{6}\Omega\rho+\frac{3}{2}\Omega p + \frac{1}{2}\Omega\mathcal{E}+\frac{1}{4}\Omega\Pi,\label{subeee1}\\
\dot{\hat{\xi}}&=-\xi\dot{\phi}-\phi\dot{\xi}+\Omega\dot{\Sigma}+\Sigma\dot{\Omega}\notag\\
&=2A\Sigma\xi-\phi\Sigma\xi-2\Omega\xi^2-\xi Q-2A\phi\Omega+\frac{3}{2}\Omega\Sigma^2\notag\\
&+\frac{1}{2}\Omega\phi^2+2\Omega^3-\Omega\mathcal{E}+\frac{1}{2}\Omega\Pi+\frac{1}{3}\Omega\left(\rho+3 p\right).\label{subeee2}
\end{align}
\end{subequations}
Taking the difference of \eqref{subeee1} and \eqref{subeee2} we obtain 

\begin{eqnarray}\label{ggh5000}
\begin{split}
\hat{\dot{\xi}}-\dot{\hat{\xi}}&=-\frac{1}{4}\phi\Sigma\xi+\frac{1}{2}\xi Q+\Omega\Sigma^2-\frac{1}{2}A\phi\Omega+\frac{1}{4}\Omega\phi^2\\
&-4\Omega^3+\frac{1}{2}\Omega\left(\rho+p\right).
\end{split}
\end{eqnarray}
Using the commutation relation in \eqref{ghh1} on \(\xi\) we have

\begin{eqnarray}\label{ggh5001}
\begin{split}
\hat{\dot{\xi}}-\dot{\hat{\xi}}&=-A\dot{\xi}+\Sigma\hat{\xi}\\
&=-\frac{1}{2}\Sigma\xi-A^2\Omega+\frac{1}{2}A\phi\Omega-\phi\Sigma\xi+\Omega\Sigma^2.
\end{split}
\end{eqnarray}
Comparing \eqref{ggh5000} and \eqref{ggh5001} we obtain the constraint

\begin{eqnarray}\label{ggh5002}
\left(A+\frac{3}{2}\phi\right)\Sigma\xi+\left(2A^2+\frac{1}{2}\phi^2+\rho+p\right)\Omega=\xi Q.
\end{eqnarray}
Let us now first prove the following proposition.

\begin{proposition}\label{propo1}
An expansion-free dynamical star that is simultaneously rotating and twisting cannot be shear-free, if it exists.
\end{proposition}

\begin{proof}
Here we assume the existence of such stars and show that if \(\Sigma=0\), then the weak energy condition (WEC) must be violated. We start by assuming that \(\Sigma=0\). Then from \eqref{ggh20} we obtain (taking into account that \(\Omega\neq 0\))

\begin{eqnarray}\label{ggh5003}
A=\frac{3}{14}\phi,
\end{eqnarray}
and therefore from \eqref{ggh23} we have

\begin{eqnarray}\label{ggh5003}
-\frac{1}{14}\phi\Omega\xi=0.
\end{eqnarray}
Since by assumption \(\xi\neq 0, \Omega\neq 0\) we must have \(\phi=0\), which implies \(A=0\) as well. Now, from \eqref{sube7} and \eqref{ggh5002} we have respectively

\begin{subequations}
\begin{align}
0&=2\xi\Omega+Q,\label{sube700}\\
\left(\rho + p\right)\Omega&=\xi Q.\label{sube701}
\end{align}
\end{subequations}
Substituting \eqref{sube700} into \eqref{sube701} and again noting that \(\Omega\neq 0\) we obtain the energy condition

\begin{eqnarray}\label{ggh5003}
\left(\rho + p\right)=-2\xi^2,
\end{eqnarray}
which gives \(\left(\rho + p\right)<0\).\qed
\end{proof}

Finally, we state and prove the following

\begin{theorem}\label{th3}
There cannot exist an expansion-free dynamical star with both rotation and spatial twist non-vanishing.
\end{theorem}
\begin{proof}
As has been shown in \cite{ssgos1}, any scalar \(\psi\) in LRS spacetimes obtained via the \(1+1+2\) decomposition satisfies the relation 

\begin{eqnarray}\label{ggghh4004}
\dot{\psi}\Omega=\hat{\psi}\xi.
\end{eqnarray}
Using \eqref{sube3} and \eqref{sube9} to substitute \(\dot{\xi}\) and \(\hat{\xi}\) for \(\dot{\psi}\) and \(\hat{\psi}\) respectively in \eqref{ggghh4004} we obtain

\begin{eqnarray}\label{ggghh4007}
-\left(2A-\phi\right)\Omega=\Sigma\xi,
\end{eqnarray}
which, upon comparing to \eqref{ggh20} gives
\begin{eqnarray}\label{ggghh4008}
\phi=\frac{10}{3}A.
\end{eqnarray}
Using \eqref{sube4} and \eqref{sube10} to substitute \(\dot{\Omega}\) and \(\hat{\Omega}\) for \(\dot{\psi}\) and \(\hat{\psi}\) respectively in \eqref{ggghh4004} we obtain

\begin{eqnarray}\label{ggghh4005}
-\phi\xi=\Omega\Sigma.
\end{eqnarray}Substituting \eqref{ggghh4008} into \eqref{ggghh4007} (or equivalently \eqref{ggh20}) we obtain 

\begin{eqnarray}\label{suber1}
\frac{2}{5}\phi\Omega=\Sigma\xi.
\end{eqnarray}
It is clear that \(\phi\neq 0\), for otherwise we would have \(A=0\) (from \eqref{ggghh4008}), in which case from \eqref{ggh20}
(or alternatively \eqref{ggghh4007}) we would have \(\xi=0\) (we have already shown that \(\Sigma\neq 0\)), contradicting the assumption that \(\xi\neq 0\).

Now, multiply both \eqref{ggghh4005} and \eqref{suber1} by \(\Omega\) we obtain respectively

\begin{subequations}
\begin{align}
-\phi\Omega\xi=\Omega^2\Sigma,\label{suber2}\\
\frac{2}{5}\phi\Omega^2=\Omega\xi\Sigma,\label{suber3}
\end{align}
\end{subequations}
which we can rewrite as

\begin{subequations}
\begin{align}
\Omega\xi=-\frac{\Omega^2\Sigma}{\phi},\label{suber4}\\
\Omega\xi=\frac{2}{5}\frac{\phi\Omega^2}{\Sigma},\label{suber5}
\end{align}
\end{subequations}
since \(\phi\neq 0,\Sigma\neq 0\). Equating \eqref{suber4} and \eqref{suber5} and simplifying we obtain

\begin{eqnarray}\label{suber6}
\left(\frac{2}{5}\phi^2+\Sigma^2\right)\Omega^2=0.
\end{eqnarray}
Since by assumption \(\Omega\neq 0\) we must have \(\left(2/5\right)\phi^2+\Sigma^2=0\), which is not possible over the set of real numbers \(\mathbb{R}\) for non-zero \(\phi\) and \(\Sigma\).\qed

\end{proof}

\section{Discussion}\label{04}

A previous paper by the same authors \cite{asr1} studied expansion-free dynamical stars in the case that the rotation and spatial twist are zero. It was shown that these stars exist under the particular conditions that the stars radiate and accelerate, and are conformally flat. As with the case of \cite{asr1} we have utilized the \(1+1+2\) semi-tetrad covariant formalism to study such stars. In this paper, we have shown that there cannot exist an expansion-free dynamical star if at least one of the rotation or the spatial twist is non-vanishing. In the case that the spatial twist is zero and the star is rotating, the star can be expansion-free but both the heat flux and the shear vanishes, in which case the energy density is time independent. Thus such expansion-free stars are not dynamical. If the rotation is zero and the spatial twist is non-vanishing then the star cannot be expansion-free, since for this to happen we must have the star being static (in which case the star is not dynamical) or the SEC must be violated. Lastly it is shown that if we assume non-vanishing of both the rotation and the spatial-twist then the shear cannot be zero. Further analysis on the basis that the shear is non-zero, using both the commutation relation and a result relating the dot and hat derivatives of an arbitrary scalar \cite{ssgos1}, show that this leads to a quadratic polynomial equation in \(\phi\) and \(\Sigma\) with no real solution for non-zero \(\phi\) and \(\Sigma\). This result, we think, is a valuable contribution to the increasing literature on the expansion-free condition. This result has also severely restricted the prevalence of such stars.

\section*{Acknowledgements}

AS acknowledges that this work is based on support from the National Research Foundation (NRF), South Africa. RG is supported by the National Research Foundation (NRF), South Africa. SDM acknowledges that this work is based on research supported by the South African Research Chair Initiative of the Department of Science and Technology and the National Research Foundation.

\section*{}


\begin{thebibliography}{99}

\bibitem{tac1} 
B. C. Tewari and K. Charan, 
\textit{}
Astrophys. Space Sci., \textbf{357}, 107 (2015).

\bibitem{tac2} 
B. C. Tewari, 
\textit{}
Gen. Relativ. Grav., \textbf{45}, 1547 (2013).

\bibitem{iv1} 
B. V. Ivanov, 
\textit{}
Int. J. Mod. Phys. D, \textbf{25}, 1650049 (2016).

\bibitem{iv2} 
B. V. Ivanov, 
\textit{}
Astrophys. Space Sci., \textbf{361}, 18 (2016).

\bibitem{iv3} 
B. V. Ivanov, 
\textit{}
Eur. Phys. J. C, \textbf{79}, 255 (2019).

\bibitem{red1} 
K. P. Reddy, M. Govender and S. D. Maharaj, 
\textit{}
Gen. Relativ. Grav., \textbf{47}, 35 (2015).

\bibitem{sant1} 
L. Herrera and N. O. Santos, 
\textit{}
Gen. Relativ. Grav., \textbf{42}, 2383 (2010).

\bibitem{gov1} 
G. Govender, M. Govender and K. S. Govinder, 
\textit{}
Int. J. Mod. Phys. D, \textbf{19}, 1773 (2010).

\bibitem{gov2} 
G. Z. Abebe, S. D. Maharaj and K. S. Govinder, 
\textit{}
Gen. Relativ. Grav., \textbf{46}, 1733 (2014).

\bibitem{gov3} 
G. Z. Abebe, S. D. Maharaj and K. S. Govinder, 
\textit{}
Eur. Phys. J. C, \textbf{75}, 496 (2015).

\bibitem{moh1} 
R. Mohanlal, S. D. Maharaj and A. K. Tiwari, 
\textit{}
Gen. Relativ. Grav., \textbf{48}, 87 (2016).

\bibitem{moh2} 
R. Mohanlal, R. Narain and S. D. Maharaj, 
\textit{}
J. Math. Phys., \textbf{58}, 072503 (2017).

\bibitem{lh1} 
L. Herrera, N.O. Santos and A. Wang,
\textit{}
Phys. Rev. D, \textbf{78}, 084026 (2008).

\bibitem{lh5} 
A. Di Prisco, L. Herrera, J. Ospino, N.O. Santos and V. M. Vina-Cervantes,
\textit{}
Int. J. Mod. Phys. D, \textbf{20}, 2351 (2011).

\bibitem{kum1} 
R. Kumar and S. K. Srivastava,
\textit{}
Int. J. Geom. Meth. Mod. Phys., \textbf{15}, 1850058 (2018).

\bibitem{kum2} 
R. Kumar and S. K. Srivastava,
\textit{}
Gen. Relativ. Grav., \textbf{50}, 95 (2018).

\bibitem{pp1} 
P. J. E. Peebles,
\textit{}
Astrophys. J., \textbf{557}, 495 (2001).

\bibitem{lh2} 
L. Herrera, G. Le Denmat and N.O. Santos,
\textit{}
Phys. Rev. D, \textbf{79}, 087505 (2009).

\bibitem{skr1} 
V. A. Skripkin,
\textit{}
Sov. Phys. Dokl, \textbf{135}, 1072 (1960).

\bibitem{lh3} 
L. Herrera, G. Le Denmat and N.O. Santos,
\textit{}
Gen. Relativ. Grav., \textbf{44}, 1143 (2012).

\bibitem{asr1} 
A. Sherif, R. Goswami and S. Maharaj,
\textit{}
Phys. Rev. D, \textbf{100}, 044039 (2019).

\bibitem{mr1} 
M. May and R. White, 
\textit{}
Phys. Rev., \textbf{141}, 1232 (1966).

\bibitem{jw1} 
J. Wilson, 
\textit{}
Astrophys. J., \textbf{163}, 209 (1971).

\bibitem{nos1} 
N. O. Santos, 
\textit{}
Mon. Not. R. Astron. Soc., \textbf{216}, 403 (1985).

\bibitem{gmm1} 
M. Govender, R. Maartens and S. D. Maharaj, 
\textit{}
Phys. Lett. A, \textbf{283}, 71 (2001).

\bibitem{rittt1} 
R. Goswami, 
\textit{}
arXiv: 0707.1122.

\bibitem{pg1} 
P. J. Greenberg, 
\textit{}
J. Math. Anal. Appl., \textbf{30}, 128 (1970).

\bibitem{ts1} 
M. Tsamparlis and D. P. Mason, 
\textit{}
J. Math. Phys., \textbf{24}, 1577 (1983).

\bibitem{cc1} 
C. Clarkson,
\textit{}
Phys. Rev. D, \textbf{76}, 104034, (2007).

\bibitem{ts2} 
D. P. Mason and M. Tsamparlis, 
\textit{}
J. Math. Phys., \textbf{26}, 2881 (1985).

\bibitem{ts3} 
M. Tsamparlis, 
\textit{} 
J. Math. Phys., \textbf{33}, 1472 (1992).

\bibitem{rit1} 
G. F. R. Ellis, R. Goswami, A. I. M. Hamid, and S. D. Maharaj,
\textit{}
Phys. Rev. D, \textbf{90}, 084013 (2014).
.
\bibitem{ggff1} 
G. F. R. Ellis,
\textit{}
J. Math. Phys., \textbf{8}, 1171 (1967).

\bibitem{sge1} 
J. M. Stewart and G. F. R. Ellis,
\textit{}
J. Math. Phys., \textbf{9}, 1072 (1968).

\bibitem{crb1} 
C. A. Clarkson and R. K. Barrett, 
\textit{}
Class. Quantum Grav., \textbf{20}, 3855 (2003).

\bibitem{gbc1} 
G. Betschart and C. A. Clarkson, 
\textit{}
Class. Quantum Grav., \textbf{21}, 5587 (2004).

\bibitem{vege1} 
H. V. Elst and G. F. R. Ellis, 
\textit{}
Class. Quantum Grav., \textbf{13}, 1099, (1996).

\bibitem{ssgos1} 
S. Singh, G. F. R. Ellis, R. Goswami and S. D. Maharaj, 
\textit{}
Phys. Rev. D, \textbf{94}, 104040 (2016).

\bibitem{ssgos2} 
S. Singh, G. F. R. Ellis, R. Goswami and S. D. Maharaj, 
\textit{}
Phys. Rev. D, \textbf{96}, 064049, (2017).

\bibitem{gft2} 
G. F. R. Ellis, 
\textit{}
J. Math. Phys., \textbf{8}, 1171, (1967).

\bibitem{ggff2} 
G. F. R. Ellis,
\textit{}
Proceedings of The International School of Physics, Course 47, Academic Press, \textbf{8}, 104 (1971).

\bibitem{rit2} 
R. Goswami and G. F. R. Ellis, 
\textit{}
Gen. Relativ. Grav., \textbf{43}, 2157, (2011).

\bibitem{gfre1} 
G. F. R. Ellis, 
\textit{}
Gen. Relativ. Grav., \textbf{41}, 581, (2009).

\end{thebibliography}
\end{document}